\documentclass[letterpaper,11pt]{article}
\usepackage{epsf}
\usepackage{graphicx,psfrag,color,pstcol,pst-grad}
\usepackage{amsthm,amsmath,amssymb}
\usepackage[]{graphicx}
\newtheorem{The}{Theorem}
\newtheorem{Pro}[The]{Proposition}
\newtheorem{Def}[The]{Definition}
\newtheorem{Ex}[The]{Example}

\begin{document}
\numberwithin{equation}{section}
\numberwithin{The}{section}

\pagestyle{empty}
\begin{center}{\Large{\bf Measures of extremality and a rank test for functional data}}
\end{center}
\begin{center}
\vspace{5mm}
{\bf{Alba M. Franco-Pereira *}}  \\
Department of Statistics and OR\\
Universidad Complutense de Madrid\\
albfranc@ucm.es

\vspace{5mm}

{\bf{Rosa E. Lillo}}\\
Department of Statistics\\
Universidad Carlos III de Madrid\\
rosaelvira.lillo@uc3m.es

\vspace{5mm}

{\bf{Juan Romo}}\\
Department of Statistics\\
Universidad Carlos III de Madrid\\
juan.romo@uc3m.es

\rule{5cm}{.5mm}
\end{center}

\vspace{1cm}

\noindent {\bf\Large Abstract}\\

The statistical analysis of functional data is a growing need in
many research areas. In particular, a robust methodology is
important to study curves, which are the output of experiments in
applied statistics. In this paper we study some new
definitions which reflect the ``extremality'' of a curve with respect to a
collection of functions, and provide natural orderings
for sample curves. Their finite dimensional versions are computationally feasible and useful for studying high dimensional observations. Thus, these extreme measures are suitable for complex observations such as microarray data and images. We show the applicability of these measures designing a rank test for functional data. This functional rank test shows different growth patterns for boys and girls when it is applied to children growth data.

\medskip


\newpage
\section{Introduction}\label{S1}
The data output sophistication in different research fields requires to advance in the statistical
analysis of complex data. In functional data analysis, each observation is a real function
$x_i(t)$, $i = 1, ..., n$, $t \in I$, where $I$ is an interval in $\mathbb{R}$. There are several reasons that make necessary
the study of functional data. In many research areas (medicine, biology, economics,
engineering), the data generating process is naturally a stochastic function. Moreover, many
problems are better approached if the data are considered as functions. For instance, if each
curve is observed at different points, a multivariate analysis would not be valid, and it is
therefore necessary to smooth the data and treat them as continuous functions defined in a
common interval. For those reasons, the analysis of functional data is one of the topics that, within
the field of statistics, is receiving a steady increasing
attention in recent years.

\medskip
Multivariate techniques such as principal components, analysis of variance and regression
methods have already been extended to a functional context; see Ramsay and Silverman (2005). A fundamental task in functional data analysis is to provide an ordering within a
sample of curves that allows the definition of order statistics such as ranks and L-statistics.
An important tool to analyze these functional data aspects is the idea of statistical depth. This concept was first introduced in the multivariate context to measure the
`centrality' or the `outlyingness' of a $d$-dimensional observation with respect
to a given dataset or a population distribution and to
generalize order statistics, ranks, and medians to higher
dimensions. Several depth definitions for multivariate data have
been proposed and analyzed by Mahalanobis (1936), Tukey (1975),
Oja (1983), Liu (1990), Singh (1991), Fraiman and Meloche (1999),
Vardi and Zhang (2000), Koshevoy and Mosler (1997) and Zuo (2003), among others. However, direct generalization of current multivariate depths to functional
data often leads to either depths that are computationally
intractable or depths that do not take into account some natural
properties of the functions, such as shape. For that reason,
several specific definitions of depth for functional data have been
introduced. See for example, Vardi and Zhang (2000), Fraiman and
Muniz (2001), Cuevas, Febrero and
Fraiman (2007), Cuesta-Albertos and Nieto-Reyes (2008), L{\'o}pez-Pintado and Romo (2009) and L{\'o}pez-Pintado and Romo (2011). The definition
of depth for curves provides us with a criteria to order the
sample curves from the center-outward (from the deepest to the
most extreme).

\medskip
In many applications, however, an important problem is to measure the `extremality' or the `outlyingness' of a curve within a set of curves, instead of its centrality. This is the case, for example, when one wishes to identify outliers within a set of sample functions. Laniado, Lillo and Romo (2010) introduced this
concept of extremality to measure the ``farness'' of a
multivariate point with respect to a data cloud or to a
distribution. And, more recently, Franco-Pereira, Lillo and Romo (2011) extended their idea to functional data introducing two
definitions of extremality which are based on the notion of `half graph' of a curve: the \emph{hyperextremality} and the \emph{hypoextremality}. In this paper we study these two extreme measures in detail and propose and application designing a rank test for functional data.

\medskip
We would like to point out here that the statistical depth, since it is a measure of centrality, it also provides a notion of farness from the center. However, the concept of depth as a measure of extremality have some shortcomings that will be explained later on the paper.

\medskip
The paper is organized as follows. We recall the concepts of \emph{hyperextremality} and \emph{hypoextremality} in Section~\ref{S2}. In
Section~\ref{S3} their corresponding finite-dimensional
versions are introduced and in Section~\ref{S4} various properties of them are derived. In Section~\ref{S5} we define
the generalized versions of \emph{hyperextremality} and \emph{hypoextremality}. In Section~\ref{S6} we explain how the extreme measures can be an alternative to depth measures through some real data examples. Finally, in Section~\ref{S7} we design a rank test for functional data based on extremality and apply it to real data sets.

\vspace{1cm}

\section{Two measures of extremality for functional data}\label{S2}

First we recall the definitions of hypograph and hypergraph. Let $C(I)$ be the space of continuous functions defined on
a compact interval $I$. Consider a stochastic process $X$ with
sample paths in $C(I)$ and distribution $F_X$. Let
$x_{1}(t),\ldots,x_{n}(t)$ be a sample of curves from $F_X$. The graph
of a function $x$ in $C(I)$ is the subset of the plane $G(x)=\{(t,x(t)): t
\in I\}$, and the hypograph ($hg$) and the hypergraph ($Hg$) of $x$
are given by
\[
\begin{array}{ll}
hg(x)= & \{(t,y) \in I \times \mathbb{R}: y \leq x(t) \}, \\
Hg(x)= & \{(t,y) \in I \times \mathbb{R}: y \geq x(t) \}.
\end{array}
\]
A natural way of establishing the extremality of a curve consists of checking the proportion of functions of the sample whose graph is in the hypograph of $x$ or in its hypergraph. Based on this idea, we introduce the two following concepts
that measure the extremality of a curve within a set of curves.
\begin{Def}
The \textbf{hyperextremality} of $x$ with respect to a set of functions
$x_{1}(t),\ldots,x_{n}(t)$ is defined as
\begin{equation*}
HEM_{n}(x)=1-\frac{\sum_{i=1}^{n}I_{\{G(x_{i})\subset
hg(x)\}}}{n}=1-\frac{\sum_{i=1}^{n}I_{\{x_{i}(t)\leq x(t), t\in I
\}}}{n}.
\end{equation*}
\end{Def}
Hence, the hyperextremality of $x$ is one minus the proportion
of functions in the sample whose graph is in the hypograph of $x$; that is, one minus the proportion
of curves in the sample below $x$. As a consequence, a curve $x$ has low hyperextremality with respect to a set of curves if many curves of the sample are below $x$.

\medskip
The population version of $HEM_{n}(x)$ is
\begin{equation*}
HEM(x,F_X)\equiv HEM(x) =1-P(G(X)\subset hg(x))=1-P(X(t)\leq x(t), t\in I).
\end{equation*}

Analogously, we can define the hypoextremality of $x$ as one minus the proportion
of functions in the sample whose graph is in the hypergraph of
$x$; that is, one minus the proportion
of curves in the sample above $x$.
\begin{Def}
The \textbf{hypoextremality} of $x$ with respect to a set of functions
$x_{1}(t),\ldots,x_{n}(t)$ is defined as
\begin{equation*}
hEM_{n}(x)=1-\frac{\sum_{i=1}^{n}I_{\{G(x_{i})\subset
Hg(x)\}}}{n}=1-\frac{\sum_{i=1}^{n}I_{\{x_{i}(t)\geq x(t), t\in I
\}}}{n}.
\end{equation*}
\end{Def}
And its population version can be given as
\begin{equation*}
hEM(x,F_X) \equiv hEM(x)=1-P(G(X)\subset Hg(x))=1-P(X(t)\geq x(t), t\in I).
\end{equation*}

It is straightforward to check that, given a curve $x$, the
larger the hyperextremality or the hypoextremality of $x$
is, the more extreme is the curve $x$. Therefore, both concepts
measure the extremality of the curves, but from a different
perspective.

\medskip
Figure~\ref{fig0} shows the hypergraph and the hypergraph of the curve $x$ within a set of five curves in the interval $I=[0,0.6]$. Immediately from the plots it is easy to compute the hyperextremality and the hypoextremality of $x$:

\begin{figure}
\centering
\[
\begin{array}{cc}
\includegraphics[height=8cm,width=8cm]{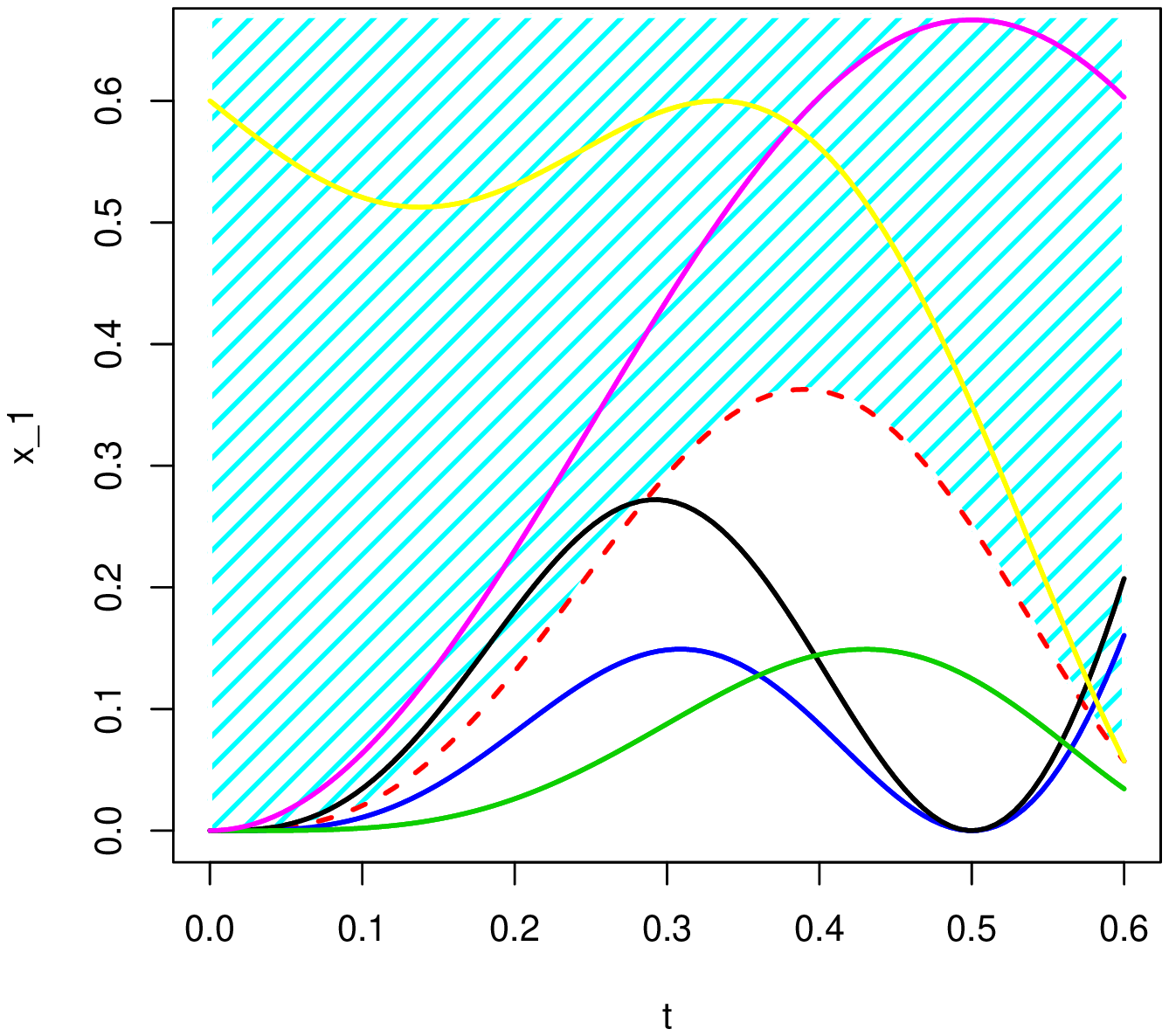} & \includegraphics[height=8cm,width=8cm]{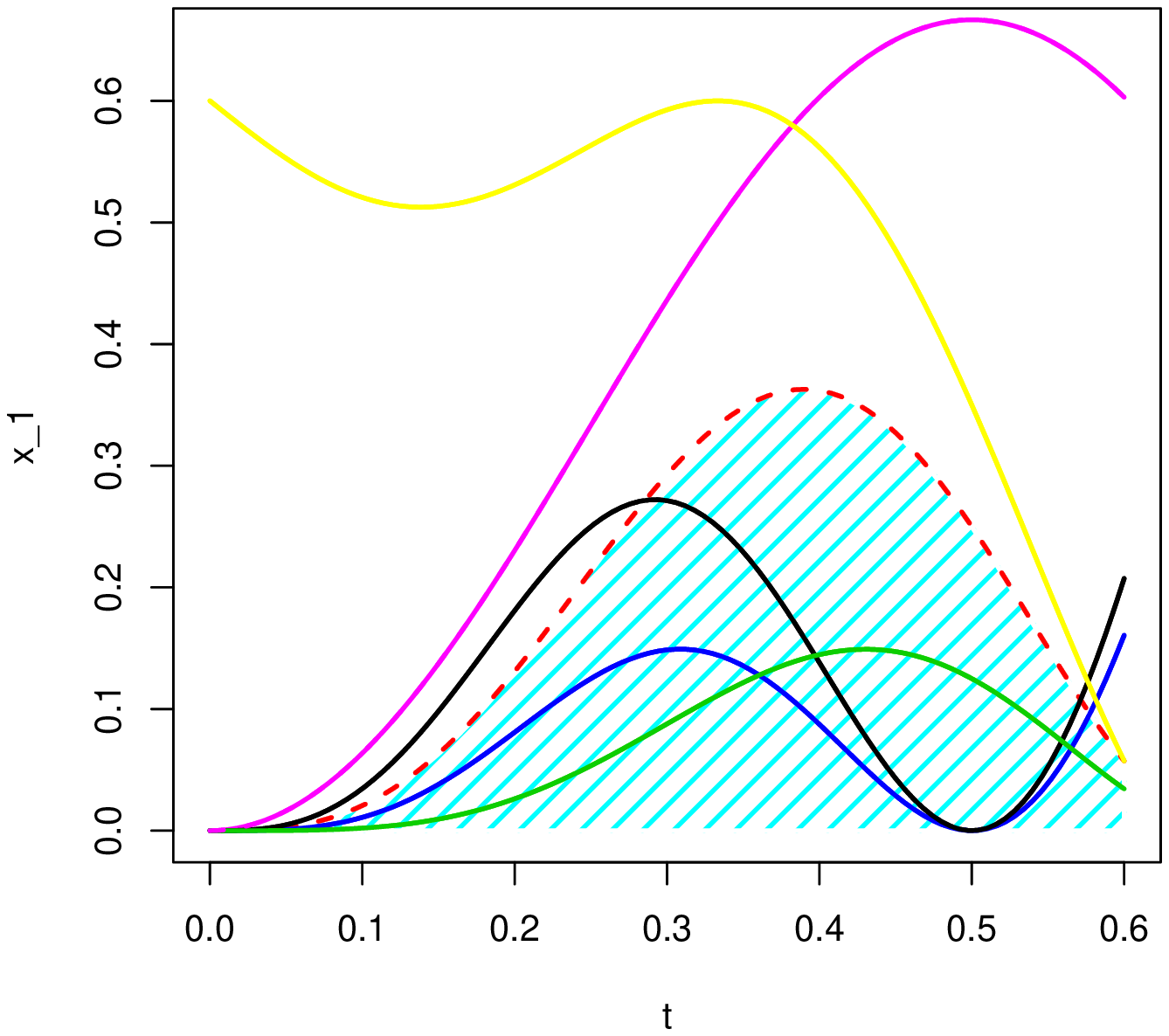}
\end{array}
\]
\caption{Hypergraph (left) and hypograph (right) of the dashed curve $x$ with respect to the set of curves}
\label{fig0}
\end{figure}

\begin{equation*}
HEM_5(x)=1-\frac{1}{5}=\frac{4}{5}=0.8,
\end{equation*}
and
\begin{equation*}
hEM_5(x)=1-\frac{2}{5}=\frac{3}{5}=0.6.
\end{equation*}

\vspace{1cm}

\section{Finite-dimensional versions}\label{S3}

The concepts of hypograph and hypergraph introduced in the
previous section can be adapted to finite-dimensional data. Consider each
point in $\mathbb{R}^{d}$ as a real function defined on the set of
indexes $\{1, ..., d\}$, the hypograph and hypergraph of a point
$x = (x(1), x(2), ..., x(d))$
can be expressed, respectively, as
\[
\begin{array}{ll}
hg(x)= & \{(k,y) \in \{1, ..., d\} \times \mathbb{R}: y \leq x(k) \}, \\
Hg(x)= & \{(k,y) \in \{1, ..., d\} \times \mathbb{R}: y
\geq x(k) \}.
\end{array}
\]

\medskip
Let $X$ be a $d$-dimensional random vector with
distribution function $F_X$. Let $X \leq
x$ and $X \geq x$ be the abbreviations
for $\{X(k) \leq x(k), k = 1, ..., d\}$ and
$\{X(k) \geq x(k), k = 1, ..., d\}$,
respectively. If we particularize the hyperextremality and the hypoextremality to the
finite-dimensional case, we obtain the following definitions:
\begin{equation*}
HEM(x,F_X)=1-P(X\leq
x)=1-F_X(x),
\end{equation*}
and
\begin{equation*}
hEM(x,F_X)=1-P(X\geq
x)=1-F_{-X}(-x)=1-F_{Y}(y),
\end{equation*}
where $Y=-X$ and $y=-x$; that is, the hyperextremality (hypoextremality) of a
$d$-dimensional point $x$ indicates the probability that a point
is componentwise greater (smaller) that $x$.

\medskip
Let $x_1, \ldots, x_n$ be a random sample from
$X$, the sample version of these extreme measures are given by
\begin{equation}\label{eq1}
HEM_{n}(x)=1-\frac{\sum_{i=1}^{n}I_{\{x_{i}\leq
x\}}}{n}=1-F_{X,n}(x),
\end{equation}
and
\begin{equation}\label{eq2}
hEM_{n}(x)=1-\frac{\sum_{i=1}^{n}I_{\{x_{i}\geq
x\}}}{n}=1-F_{Y,n}(y),
\end{equation}
where $F_{X,n}$ and $F_{Y,n}$ stand for the empirical distribution functions of $x_1,\ldots,x_n$ and $-x_1,\ldots,-x_n$, respectively.

\begin{Ex}
Let $x_1=(2,1,1)$, $x_2=(4,3,2)$ and $x_3=(6,5,5) \in \mathbb{R}^3$. Figure~\ref{fig1} illustrates the parallel coordinates (see Inselberg, 1985) of $x_0=(4.5,2,4)$ with respect to these three points. Hence, we can compute
\[
HEM_3(x_0)=1-\frac{5}{9}=\frac{4}{9},
\]
and
\[
hEM_3(x_0)=1-\frac{4}{9}=\frac{5}{9}.
\]
\end{Ex}

\begin{figure}
\centering
\includegraphics[height=8cm,width=8cm]{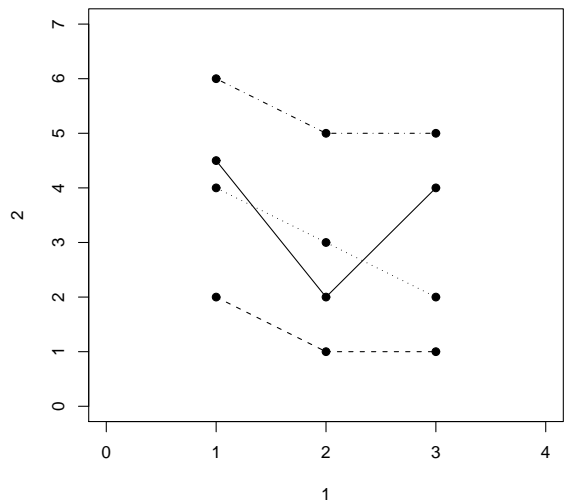}
\caption{Parallel coordinates of $x_0=(4.5,2,4)$ (solid) and $x_1=(2,1,1)$ (dashed), $x_2=(4,3,2)$ (dotted) and $x_3=(6,5,5)$ (dotdash) considered in order to compute the HEM and the hEM of $x_0$}
\label{fig1}
\end{figure}

Now, recall the definition of \textbf{oriented sub-orthant} from Laniado, Lillo and Romo (2010):
\begin{Def}
Given a unit vector $u \in \mathbb{R}^{n}$ and a vertex $x \in \mathbb{R}^{n}$, an oriented sub-orthant $C_x^{u}$ is the convex cone given by
\[
C_x^{u}=\{ z \in \mathbb{R}^{n}:Q_1Q_2'(z-x)\geq \theta\},
\]
where $\theta$ is the zero vector in $\mathbb{R}^n$ and $Q_1$ and $Q_2$ verify
\[
e=Q_1R_1 \quad \textup{and} \quad u=Q_2R_2,
\]
with $e=\frac{1}{\sqrt{n}}(1,...,1)' \in \mathbb{R}^{n}$ and $R_1=R_2=(1,0,...,0)' \in \mathbb{R}^{n}$.
\end{Def}

We can think of $C_x^{u}$ being the convex cone with vertex $x$
obtained by moving the nonnegative orthant and translating the
origin to $x$. It is easy to see that the population finite dimensional version of the
hyperextremality (hypoextremality) can be also seen as the probability that the
vector $x$ belongs to $C_x^{u}$ if
$u=\frac{1}{\sqrt{2}}(1,1)'$ ($u=\frac{1}{\sqrt{2}}(-1,-1)'$). Therefore, the hyperextremality and the hypoextremality coincide
with the extreme measure for multivariate data introduced by
Laniado, Lillo and Romo (2010), which is computationally feasible
and useful for studying high dimensional observations. 

\medskip
The hyperextremality and the hypoextremality in the finite dimensional case are invariant with respect to translation and some types of dilatations. Let $A$ be a positive (or negative) definite diagonal matrix and $b \in \mathbb{R}^{d}$, then
\[
EM(Ax+b,F_{Ax+b})=EM(x,F_X).
\]

In the following propositions we establish some other properties of these notions of extremality. The first one states that the hyperextremality (hypoextremality) increases to one when the norm of the point tends to infinity and in the second one we prove the uniform convergence of $HEM_{n}$ and $hEM_{n}$ to their corresponding population versions.
\begin{Pro}\label{luego}
Let $x \in \mathbb{R}^{d}$,
\[
\sup_{||x|| \geq M}EM(x)\rightarrow 1, \quad \textup{ when } M\rightarrow \infty,
\]
where $EM= HEM$ or $EM= hEM$.
\end{Pro}
The proof of Proposition~\ref{luego} is postponed to the next section, since it is a particular case of the same property in the functional case.

\begin{Pro}
$EM(\cdot)$ is uniformly consistent:
\[
\sup_{x \in \mathbb{R}^{d}} |EM_{n}(x)-EM(x)|\rightarrow^{a.s.}0, \quad \textup{ when } n\rightarrow \infty,
\]
where $EM= HEM$ or $EM= hEM$.
\end{Pro}

\begin{proof}
Applying Glivenko-Cantelli's theorem in $\mathbb{R}^d$, we have that
\[
\sup_{x \in \mathbb{R}^{d}} |F_{X,n}(x)-F_X(x)|\rightarrow^{a.s.}0, \quad \textup{ when } n\rightarrow \infty.
\]
Therefore,
\[
\sup_{x \in \mathbb{R}^{d}} |HEM_n(x)-HEM(x)|=\sup_{x \in \mathbb{R}^{d}} |1-F_{X,n}(x)-(1-F_X(x))|=
\]
\[
=\sup_{x \in \mathbb{R}^{d}} |F_X(x)-F_{X,n}(x)|\rightarrow^{a.s.}0, \quad \textup{ when } n\rightarrow \infty.
\]
The proof for $hEM_n$ is analogous.
\end{proof}

\section{Properties of the functional extremality measures}\label{S4}
Here we extend some of the properties established in the previous section to the functional version of extremality. Let $x_1,\ldots,x_n$ be independent copies of a stochastic process $X$ in $C(I)$ with distribution function $F_X$. Assume that the stochastic process $X$ is tight, i.e.,
\[
P(||X||_{\infty}\geq M)\rightarrow^{a.s.}0, \quad \textup{ when } M\rightarrow \infty.
\]

The two extremality measures defined in Section~\ref{S2} verify a linear invariance property. Consider $a$ and $b$ functions in $C(I)$, where $a(t)>0$ or $a(t)<0$ for every $t \in I$. Then,
\[
EM(x,F_X)=EM(ax+b,F_{aX+b}),
\]
where $EM= HEM$ or $EM= hEM$.

The hyperextremality and hypoextremality of a function converge to one when its norm tends to infinity.
\begin{Pro}
The extremality measures verify that
\[
\sup_{||x||_{\infty} \geq M} EM(x,F_{X})\rightarrow 1, \quad \textup{ when } M\rightarrow \infty,
\]
and
\[
\sup_{||x||_{\infty} \geq M} EM_n(x,F_{n,X})\rightarrow^{a.s.} 1, \quad \textup{ when } M\rightarrow \infty,
\]
where $EM= HEM$ and $EM_n= HEM_n$, or $EM= hEM$ and $EM_n= hEM_n$.
\end{Pro}

\begin{proof}
The quantity $\sup_{||x||_{\infty} \geq M} EM(x,F_{X})$ can be decomposed depending on where the supremum is achieved in the following way:
\[
\sup_{||x||_{\infty} \geq M} EM(x,F_{X}) \leq \sup_{||x||_{\infty} = M \bigcap ||x||_{\infty}=\sup x(t)} EM(x,F_{X}) +
\]
\[
+ \sup_{||x||_{\infty} \geq M \bigcap ||x||_{\infty}=\sup -x(t)} EM(x,F_{X}).
\]
Now,
\[
\sup_{||x||_{\infty} \geq M \bigcap ||x||_{\infty}=\sup x(t)} EM(x,F_{X}) = \sup_{||x||_{\infty} \geq M \bigcap ||x||_{\infty}=\sup x(t)} (1-P(X(t)\geq x(t))) =
\]
\[
=1- \sup_{||x||_{\infty} \geq M \bigcap ||x||_{\infty}=\sup x(t)} P(X(t)\geq x(t)) \geq 1-\sup_{||x||_{\infty} \geq M \bigcap ||x||_{\infty}=\sup x(t)} P(||X||_{\infty}\geq ||x(t)||_{\infty}) \geq
\]
\[
\geq 1- P(||X||_{\infty}\geq M)\rightarrow 1, \quad \textup{ when } M\rightarrow \infty.
\]
And also,
\[
\sup_{||x||_{\infty} \geq M \bigcap ||x||_{\infty}=\sup -x(t)} EM(x,F_{X}) = \sup_{||x||_{\infty} \geq M \bigcap ||x||_{\infty}=\sup -x(t)} (1-P(X(t)\leq x(t)))=
\]
\[
=1- \sup_{||x||_{\infty} \geq M \bigcap ||x||_{\infty}=\sup -x(t)} P(X(t)\leq x(t)) \geq 1-\sup_{||x||_{\infty} \geq M \bigcap ||x||_{\infty}=\sup -x(t)} P(-X(t) \leq -x(t)) \geq
\]
\[
\geq 1-\sup_{||x||_{\infty} \geq M \bigcap ||x||_{\infty}=\sup -x(t)} P(||-X(t)||_{\infty} \leq ||x(t))||_{\infty}
\geq 1- P(||-X||_{\infty}\geq M)\rightarrow 1, \quad \textup{ when } M\rightarrow \infty.
\]

\medskip
To prove that $EM_n(x,F_{n,X})$ converges almost surely to one we can use the same decomposition as before. Hence, here we just present a sketch of the proof. If $||x||_{\infty=\sup x(t)}$,
\[
\sup_{||x||_{\infty} \geq M \bigcap ||x||_{\infty}=\sup x(t)} EM(x,F_{n,X}) \geq \sup_{||x||_{\infty} \geq M \bigcap ||x||_{\infty}=\sup x(t)} (1-\frac{1}{n} \sum_{i=1}^{n}I_{\{X_i(t)\geq x(t), t \in I\}}) =
\]
\[
=1- \sup_{||x||_{\infty} \geq M \bigcap ||x||_{\infty}=\sup x(t)} \frac{1}{n} \sum_{i=1}^{n}I_{\{X_i(t)\geq x(t), t \in I\}} \geq 1-\sup_{||x||_{\infty} \geq M \bigcap ||x||_{\infty}=\sup x(t)}
\frac{1}{n} \sum_{i=1}^{n}I_{\{||X_i||_{\infty} \geq ||x||_{\infty}\}} \geq
\]
\[
\geq 1-\frac{1}{n} \sum_{i=1}^{n} \sup_{||x||_{\infty} \geq M}  I_{\{||X_i||_{\infty} \geq ||x||_{\infty}\}}.
\]

\medskip
In what follows we show that $X_M=\sup_{||x||_{\infty}\geq M} I_{\{||X_i||_{\infty} \geq ||x||_{\infty}\}}$ converges almost surely to $0$ when $M$ tends to infinity. Define $Y_M=I_{\{||X_i||_{\infty}\geq M\}}$, since
\[
0\leq X_M \leq Y_M,
\]
it is sufficient to prove that $Y_{M}\rightarrow^{a.s.} 0$, or equivalently that
\[
P(\sup_{M \geq l} I_{\{||X_i||_{\infty} \geq M\}} > \varepsilon) \rightarrow 0, \quad \textup{ when } l\rightarrow \infty.
\]
It is easy to see that the following inequality holds,
\[
\sup_{M \geq l} I_{\{||X_i||_{\infty} \geq M\}} \leq I_{\{ ||X_i||_{\infty} \geq l\}},
\]
and it implies that
\[
P(\sup_{M \geq l} I_{\{||X_i||_{\infty} \geq M\}} > \varepsilon) \leq P(I_{\{ ||X_i||_{\infty} \geq l \} \varepsilon})
=P(||X_i||_{\infty} \geq l)\rightarrow 0, \quad \textup{ when } l\rightarrow \infty.
\]
\end{proof}

The following theorem deals with the consistency of our extremality measures in for functional setting.
\begin{The}
$EM_n$ is strongly consistent.
\[
EM_n(x,F_{n,X}) \rightarrow^{a.s.} EM(x,F_X),
\]
where $EM= HEM$ and $EM_n= HEM_n$, or $EM= hEM$ and $EM_n= hEM_n$.
\end{The}

\begin{proof}
This results is a consequence of the law of large numbers.
\end{proof}

\section{Generalized extreme measures}\label{S5}
The hyperextremality and the hypoextremality are useful tools to measure the `farness' of a curve with respect to a set of functions. If the curves are regular (in the sense of their shape) then these two measures allow us to identify those that are far from the center. That is, they provide us with a tool to detect those curves that are different from the rest in magnitude. However, when the bunch of curves are very irregular, it is possible that these measures give always numbers very close to one, since there are not many curves lying totally below or above them. In such cases it is convenient to use modified definitions of extremality. Here we introduce the generalized versions of hyperextremality and hypoextremality, less restrictive than the previous versions and that are suitable for irregular curves.

\medskip
We define the modified or generalized
versions of both extreme measures considering the Lebesgue's
measure instead of the indicator function in (\ref{eq1}) and (\ref{eq2}).

\begin{Def}
The \textbf{generalized hyperextremality} (MHEM) and the \textbf{generalized hypoextremality} (MhEM) of $x$ with respect to a set of
functions $x_{1}(t),\ldots,x_{n}(t)$ are, respectively,
\begin{equation*}\label{eq8.2.1}
MHEM_{n}(x)=1-\sum_{i=1}^{n}\frac{\lambda({\{G(x_{i})\subset
hg(x)\}})}{n\lambda(I)},
\end{equation*}
and
\begin{equation*}\label{eq8.2.3}
MhEM_{n}(x)=1-\sum_{i=1}^{n}\frac{\lambda({\{G(x_{i})\subset
Hg(x)\}})}{n\lambda(I)},
\end{equation*}
where $\lambda$ stands for the Lebesgue's measure on $\mathbb{R}$.
\end{Def}
Hence, the generalized hyperextremality (hypoextremality) of $x$ is one minus the
``proportion of time'' that the graphs of the functions of the
sample are in the hypograph (hypergraph) of $x$. That is, the proportion of
time that the curves of the sample are above (below) $x$.

\medskip
Now, if $C_x^{\overrightarrow{u}}$ be a convex cone with vertex $x$
obtained by moving the nonnegative orthant and translating the
origin to $x$. Then, the finite dimensional version of the
generalized hyperextremality (hyporextremality) can be also seen as the proportion of
coordinates of $x$ that belongs to $C_x^{u}$
where $u=\frac{1}{\sqrt{2}}(1,1)'$ ($u=\frac{1}{\sqrt{2}}(-1,-1)'$).

\medskip
The population versions of $MHEM_{n}(x)$ and $MhEM_{n}(x)$ are
\begin{equation*}\label{eq8.2.2}
MHEM(x,P)=1-\frac{E(\lambda({\{t \in I: x(t) \leq X(t)\}}))}{\lambda(I)},
\end{equation*}
and
\begin{equation*}\label{eq8.2.4}
MhEM(x,P)=1-\frac{E(\lambda({\{t \in I: x(t) \geq X(t)\}}))}{\lambda(I)}.
\end{equation*}

Following the notation of L\'opez-Pintado and Romo (2011) the above definitions can be expressed in terms of the superior and inferior lengths.

\medskip
Using the hyperextremality (hypoextremality) or the generalized
hyperextremality (hypoextremality) depends on the kind of
functions being analyzed and the objectives to be checked. If the
curves are very irregular, it is convenient to use the modified
versions because it avoids having too
many ties and there might not be a representative `shape'. The
hyperextremality and the hypoextremality are more adequate if the curves
are smooth in terms of shape.

\medskip
Going back to the example shown in Figure~\ref{fig0}, we can compute the generalized
hyperextremality  and the generalized hypoextremality of the dashed curve $x$ just taking into account the length of the intervals in which the curve is below or above any of the other curves. Thus,
\begin{equation*}
MHEM_n(x)=1-\frac{0.6+0.575+0.3}{6*0.6}=0.5903,
\end{equation*}
and
\begin{equation*}
MhEM_n(x)=1-\frac{0.6+0.6+0.28+0.25+0.26}{6*0.6}=0.4472.
\end{equation*}

\section{Extremality measures as an alternative to statistical depth}\label{S6}
One of the direct applications of these measures is to define the frontier from which we can decide whether a function is an outlier or not, since these definitions provide a natural ordering for functions. Note that the concept of statistical depth also provides a ordering for functional data but this is a center-outward ordering. The idea of extremality is similar to the half-graph depth defined in L{\'o}pez-Pintado and Romo (2011) but in order to measure the extremality of a curve this is a more natural approach. Therefore, these measures provide an alternative ordering which is more natural to that induced by depth definitions (center-outward). Besides, these measures are computationally very fast compared to the statistical depth.

\medskip
Another interesting application, that reveals some shortcomings of the use of statistical depth, is described next. Recently, Franco-Pereira, Lillo and Romo (2012) presented a statistical tool to construct confidence bands for the difference of two percentile residual life functions It is based on bootstrap techniques and the use of depth for functions. This methodology, though computationally slow, is useful to have an idea of whether two random variables are ordered with respect to the percentile residual life order. However, this approach is not valid to conduct an hypothesis test of the form $H_0:q_{X}(t)\leq q_{Y}(t)$ for all $t \in I$ versus $H_1: \textup{ there exists }t' \in I \textup{ such that } q_{X}(t')> q_{Y}(t')$. The reason is that, given $\alpha \in (0,1)$, the proposed $(1-\alpha)$-confidence band is given by the hull of the $(1-\alpha)\cdot 100 \%$ deepest curves (more central curves) in the sample. Since our goal is to test that the two curves are equal, we would reject the null hypothesis if the lower limit of the band is above the $x$-axis for some $t \in I$. However, the approach followed to construct the confidence bands does not guarantee that we leave the same probability in both sides outside each band. What we know is that outside the band there is a probability of $\alpha$ but this probability is not $\alpha/2$ above the band and $\alpha/2$ below the band. A new methodology based on the extremality measures do take this into account so it means a natural alternative to the concept of statistical depth in such contexts is possible.

\medskip
\textbf{Ramsay and Silverman (2005)}
\textbf{Canadian Weather} 
In many fields of environmental sciences such as agronomy, ecology, meteorology
or monitoring of contamination and pollution, the observations consist of samples
of random functions.

\medskip
Here we use a well-known meteorological data set in FDA consisting of daily temperature and precipitation measurements recorded
at 35 weather stations of Canada (Ramsay and Dalzell (1991) and Ramsay and Silverman (2005)). These authors use Fourier basis functions for constructing curves
from discrete data. They apply functional principal components and functional
linear models to describe the modes of variability in temperature curves, and for
establishing the influence of temperature on precipitation. We specifically use the
temperature values of this data set to provide an applied context for our proposal.
In particular we analyze information of daily temperature averaged over the years
1960 to 1994 (February 29th combined with February 28th). See Figure~\ref{curvas}.
The data for each station were obtained from Ramsay and Silverman's home page
(http://www.functionaldata.org/).

\medskip
We have considered the average daily temperature for each day of the year. In Figure~\ref{curvasHEM}, we compute the $80\%$ more central curves leaving outside the $10\%$ of the curves with higher hyperextremality and the $10\%$ of the curves with higher hypoextremality. In Figure~\ref{curvasMHEM}, we compute the $80\%$ more central curves leaving outside the $10\%$ of the curves with higher generalized hyperextremality and the $10\%$ of the curves with higher generalized hypoextremality. In this example, since the curves are quite regular in terms of shape, there is no a big different visually. However we can appreciate that the result using both is different. The last one reflects better the intuitive idea of centrality and extremality.

\section{Rank tests for functional data}\label{S7}
The extremality definitions for curves allow us to extend the rank test to functional data. Liu and Singh (1993) generalized to multivariate data the univariate Wilcoxon rank test through the order induced by a multivariate depth and L{\'o}pez-Pintado and Romo (2009) generalized this idea to functional data. Here we present an alternative rank test to this one. The advantages of our proposal is that the extremality measures we present here are computationally faster and that the order induced is more natural than the order induced by the statistical depth for functions since the former provides a center-outward ordering.

\medskip
Following Liu and Singh (1993) approach, let $x_1,...,x_n$ be a sample of curves and let $P(=P_{n})$ be its empirical distribution. We define
\[
R(P,x_i)=\{ \textup{proportion of $x_j$'s from the sample with } EM_n(x_j) \geq EM_n(x_i)\},
\]
where $EM_n$ can be either the sample hyperextremality, the sample hypoextremality or their corresponding generalized versions. Note that $R(P,x_i)$ takes values between 0 and 1. We rank the observations according to the increasing values of $R(P,x_i)$, assigning them an integer from 1 to $n$. If there are curves with the same value of $R$, $R(P,x_{i_1})=R(P,x_{i_2})=...=R(P,x_{i_j})$, with $i_1 < i_2 <...<i_j$, we consider the rank of $x_{i_{k+1}}$ as the rank of $x_{i_{k}}$ plus one. We propose a test based on these ranks to decide if two groups of curves come from the same population.

\medskip
Let $x_1,x_2,...,x_n$ be a sample of curves from population $P_1$ and let $y_1,y_2,...,y_m$ be a sample of curves from population $P_2$. Assume that there is a third reference sample $Z=\{z_1,z_2,...,z_{n_0}\}$ from one of the two populations, for example $P_1$, with $n_0$ greater than $n$ and $m$. Let $P_{0}$ be the corresponding empirical distribution. Calculate $R(P_0,x_i)=\textup{proportion of $z_j$'s with }EM_n(z_j,P_{0}) \geq EM_n(x_i,P_{0})$, and $R(P_{0},y_i)=\textup{proportion of $z_j$'s with } EM_n(z_j,P_{0}) \geq EM_n(y_i,P_0)$. They express the position of each $x_i$ and $y_i$ with respect to $Z$. Order this values $R(P_0,x_i)$ and $R(P_0,y_i)$ from smallest to highest giving them a rank from 1 to $n+m$. If there are ties, we apply the previous criterion. The proposed statistic to test $H_0:P_1=P_2$ is $W=\textup{The sum of the ranks of }R(P_0,y_j)$. The ranks of $R(P_0,y_j)$ behave under $H_0$ as $m$ numbers randomly chosen from \{1,2,...,n+m\}. Hence, the distribution of $W$ is the distribution of $\rho_1+\rho_2+...+\rho_m$ where $\rho_1,\rho_2,...,\rho_m$ is a sample without replacement of \{1,2,...,n+m\} (see Liu and Singh (1993)). The null hypothesis is rejected when $W$ is small, because it indicates that the distributions are not the same.

\medskip
We have applied this test to real data. First we have considered the growth curves for boy and girls (see Ramsay and Silverman, 2005). We have applied rank test to decide if there are no differences between both groups curves. See Figure~\ref{figGROWTH}. The $p$-value with $HEM$ is 0.02414; hence, we reject the null hypothesis at the 0.05 significance level, concluding that there exit significant differences between the growth curves for boys and girls. When we consider the $hEM$, the $MHEM$ and the $MhEM$, instead of the $HEM$, the corresponding $p$-values are 0.04317, 8.287$e^{-06}$ and 1.408$e^{-06}$. Again, we conclude that there are different growth patterns for boys and girls.


\section{Conclusions}\label{S8}
We have studied the notions of hyperextremality and hypoextremality, first introduced in Franco-Pereira, Lillo and Romo (2011). They reflect the ``extremality'' of a curve with respect to a
collection of functions and provide natural orderings
for sample curves. Therefore, these measures happen to be a natural alternative to the statistical depth for functions.

\medskip
We have also designed a new rank test for functional data based on these notions. This functional rank test shows different growth patterns for boys and girls when it is applied to children growth data.

\bigskip

\noindent {\bf\Large Acknowledgements}\\

The research of Alba M. Franco-Pereira is supported by the project MTM2011 of the Spanish Ministerio de Ciencia e Innovaci\'{o}n (FEDER support included).


\begin{figure}
\centering
\includegraphics[height=8cm,width=8cm]{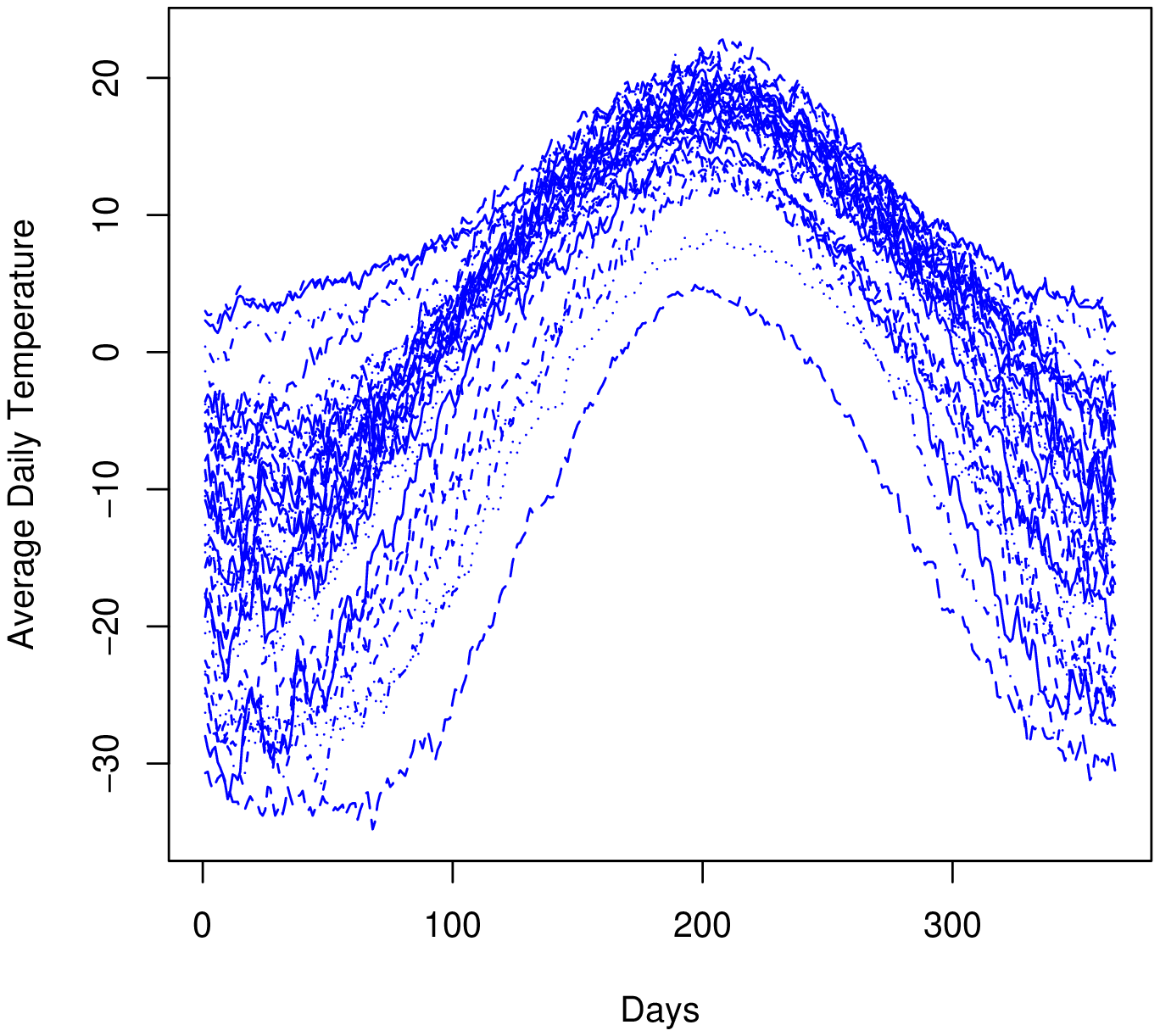}
\caption{Data}
\label{curvas}
\end{figure}

\begin{figure}
\centering
\[
\begin{array}{ccc}
\includegraphics[width=.30\textwidth, height=.315\textwidth]{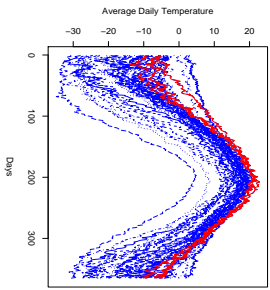} & \includegraphics[width=.30\textwidth, height=.3\textwidth]{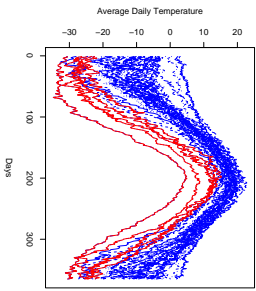} & \includegraphics[width=.30\textwidth, height=.3\textwidth]{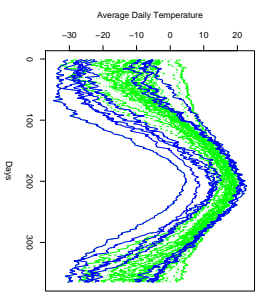}
\end{array}
\]
\caption{a) $10\%$ curves with higher HEM (red), b) $10\%$ curves with higher hEM (red), c) $80\%$ ``less extreme'' curves (green)}
\label{curvasHEM}
\end{figure}

\begin{figure}
\centering
\[
\begin{array}{ccc}
\includegraphics[width=.30\textwidth, height=.315\textwidth]{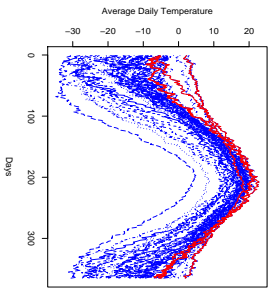} & \includegraphics[width=.30\textwidth, height=.3\textwidth]{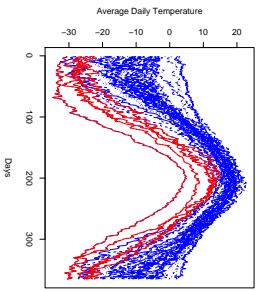} & \includegraphics[width=.30\textwidth, height=.3\textwidth]{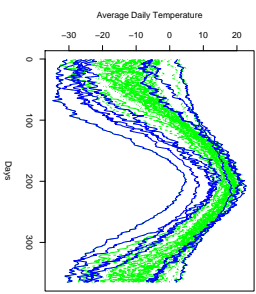}
\end{array}
\]
\caption{a) $10\%$ curves with higher MHEM (red) b) $10\%$ curves with higher MhEM (red), c) $80\%$ ``less generalized extreme'' curves (green)}
\label{curvasMHEM}
\end{figure}
%
\begin{figure}
\centering
\includegraphics[height=8cm,width=8cm]{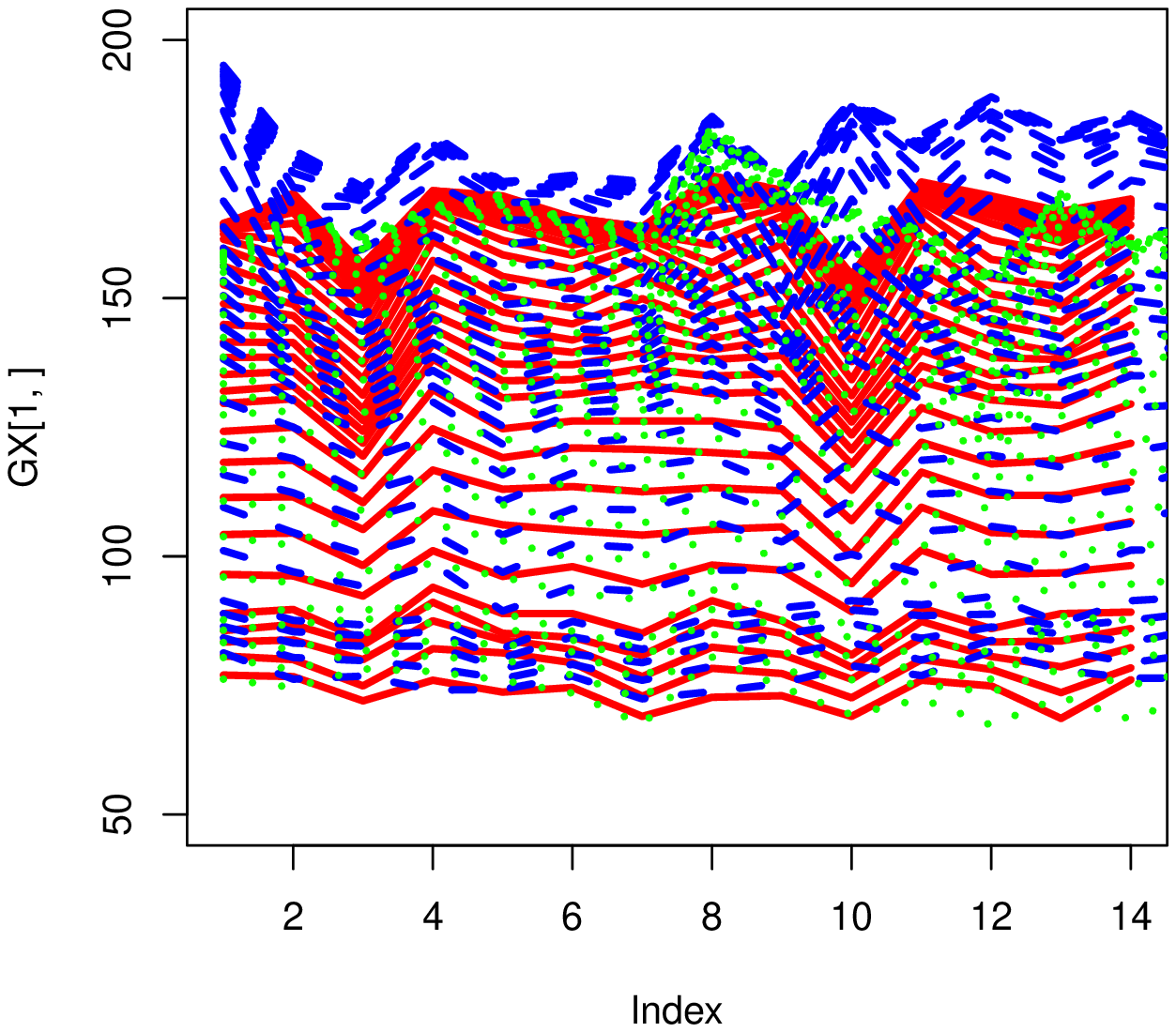}
\caption{Three groups of curves}
\label{figGROWTH}
\end{figure}

\begin{thebibliography}{9}
\bibitem{C-A}Cuesta-Albertos, J. and Nieto-Reyes, A. (2008). The random Tukey depth. \textsl{Computational Statistics and Data Analysis} \textbf{52}, 4979--4988.
\bibitem{CFF}Cuevas, A., Febrero, M. and Fraiman, R. (2007). Robust estimation and classification for functional data via projection-based depth notions. \textsl{Computational Statistics} \textbf{22}, 481--496.
\bibitem{CF}Cuevas, A. and Fraiman, R. (2009). On depth measures and dual statistics. A methodology for dealing with general data. \textsl{Journal of Multivariate Analysis} \textbf{100}, 753--766.
\bibitem{FM}Fraiman, R. and Meloche, J. (1999). Multivariate $L$-estimation. \textsl{Test} \textbf{8}, 255--317.
\bibitem{FM2}Fraiman, R. and Muniz, G. (2001). Trimmed means for functional data. \textsl{Test} \textbf{10}, 419--440.
\bibitem{FP}Franco-Pereira, A. M., Lillo, R. E., and Romo, J. (2012). Comparing quantile residual life functions by confidence bands.  \textsl{Lifetime data analysis} \textbf{18}, 195--214..
\bibitem{FP}Franco-Pereira, A. M., Lillo, R. E., and Romo, J. (2011). Extremality for functional data. In \emph{Recent advances in functional data analysis and related topics}. Springer, New York.
\bibitem{I}Inselberg (1985). The plane with parallel coordinates. Invited paper. \textsl{Visual Computer} \textbf{1},
69--91.
\bibitem{KM}Koshevoy, G. and Mosler, K. (1997). Zonoid trimming for multivariate distributions. \textsl{The Annals of Statistics} \textbf{25}, 1998--2017.
\bibitem{LLR}Laniado, H., Lillo, R. E., and Romo, J. (2010). Multivariate extremality measure. \textsl{Working paper} \textbf{10-19}. Statistics and Econometrics Series
08. Universidad Carlos III de Madrid.
\bibitem{Liu}Liu, R. (1990). On a notion of data depth based on random simplices. \textsl{The Annals of Statistics} \textbf{18}, 405--414.
\bibitem{LS}Liu, R. and Singh, K. (1993). A quality index based on data depth and multivariate rank test. \textsl{Journal of the American Statistical Association} \textbf{88}, 257--260.
\bibitem{LR}L{\'o}pez-Pintado, S. and Romo, J. (2009). On the concept of depth for functional data. \textsl{Journal of the American Statistical Association} \textbf{104}, 718--734.
    \bibitem{LR05}L{\'o}pez-Pintado, S. and Romo, J. (2011). A half-graph depth for functional data. \textsl{Computational Statistics and Data Analysis}, to appear.
\bibitem{M}Mahalanobis, P. C. (1936). On the generalized distance in statistics. \textsl{Proceedings of National Academy of Science of India} \textbf{12},
49--55.
\bibitem{O}Oja, H. (1983). Descriptive statistics for multivariate distributions. \textsl{Statistics and Probability Letters} \textbf{1}, 327--332.
\bibitem{RD}Ramsay, J. O. and Dalzell, C. J. (1991). Some tools for functional data
analysis. \emph{Journal of the Royal Statistical Society, Series B}, \textbf{53}, 539--572.
\bibitem{RS2}Ramsay, J. O. and Silverman, B. W. (2002). Applied Functional Data Analysis; Methods and Case Studies. Springer, New York.
\bibitem{RS}Ramsay, J. O. and Silverman, B. W. (2005). \textsl{Functional data analysis}. Second Edition. Springer-Verlag.
\bibitem{Si}Singh, K. (1991). A notion of majority depth. Unpublished document. 
\bibitem{T}Tukey, J. (1975). Mathematics and picturing data. \textsl{Proceedings of the 1975 International Congress of Mathematics} \textbf{2}, 523--531.
\bibitem{VZ}Vardi, Y. and Zhang, C. H. (2000). The multivariate $L_{1}$-median and associated data depth. \textsl{Proceedings of the National Academy of Science USA} \textbf{97}, 1423--1426.
\bibitem{Z}Zuo, Y. (2003). Projection based depth functions and associated medians. \textsl{The Annals of Statistics} \textbf{31}, 1460--1490.
\end{thebibliography}
\end{document}